\documentclass[10pt]{amsart}
\usepackage{inputenc}
\usepackage[T2A]{fontenc}
\inputencoding{cp1251}

\usepackage{enumerate}
\usepackage{amsmath,amsthm,amssymb}
\usepackage{amsfonts}
\usepackage{amsthm}
\usepackage{epsfig}

\newtheorem{thm}{Theorem}[section]
\newtheorem{lem}{Lemma}[section]
\newtheorem{cor}{Corollary}[section]

\theoremstyle{definition}
\newtheorem*{Def*}{Definition}
\newtheorem*{rems*}{Remarks}
\numberwithin{equation}{section}

\newcommand{\D}{\displaystyle}
\newcommand{\be}{\begin{equation}}
\newcommand{\ee}{\end{equation}}

\newcommand{\ii}{\textup{i}}
\newcommand{\la}{\lambda}

\begin{document}
\title[Periodic problem for the nonlinear Schr\"odinger  equation]{Periodic problem for the nonlinear Schr\"odinger  equation}

\author{ V.P. Kotlyarov}
\address{Institute for Low Temperature Physics\\ 47,Lenin ave\\ 61103 Kharkiv\\ Ukraine}
\email{kotlyarov@ilt.kharkov.ua}
\bigskip

\centerline{\bf Foreword}
{\sl These are the English translation of the papers published only on Russian:\\
{\rm V.P. Kotlyarov, Periodic problem for the Schr\"odinger nonlinear equation,}\\
In: Voprosy matematicheskoi fiziki i funkcionalnogo analiza, {\bf 1}, Naukova Dumka, Kiev, 1976, pp.121-131 and \\
{\rm A.R.Its, V.P.Kotlyarov, Explicit formulas for solutions of the Schr\"odinger nonlinear equation,}\\  "Doklady Akad.Nauk Ukrainian SSR", ser.A, №10, 1976, pp.965-968.\\\\
During the last time I have received some requests on these papers.
I hope this reprint will be useful for interested readers (V.Kotlyarov).\\\\
}

\begin{abstract}
The paper offers the method of discovering of some class of solutions for  the nonlinear Schr\"odinger  equation $$\ii u_t+u_{xx}+2|u|^2u=0.$$
An algorithm of constructive solving of the Cauchy periodic problem with a finite-gap initial condition was also  obtained.
\end{abstract}
\maketitle

\section{Introduction}

In \cite{GGKM} was discovered  a method of the construction of decreasing solutions of the Korteweg de Vries (KdV) equation. The method bases on the inverse scattering theory for the Sturm-Lioville operator. Recently this method was developed \cite{M}-\cite{D} for obtaining periodic solutions to the KdV equation.

In the present paper we investigate the periodic finite-gap problems   for  the nonlinear  Schr\"odinger  equation
\be\label{NLS}
\ii v_t+v_{xx}+2|v|^2v=0,
\ee
by using  a method from \cite{M}.

The periodic finite-gap problem for another type of nonlinear equation
\be\label{NLS1}
\ii v_t+v_{xx}-2|v|^2v=0,
\ee
is studied in \cite{I0}, \cite{I}. In this case (equation (\ref{NLS1})), the spectral problem is connected with the Dirac self-adjoint operator while for equation (\ref{NLS}) the spectral problem  is not self-adjoint.

For the first time these nonlinear equations were studied in \cite{ZSh}, where the Cauchy problem was solved for the decreasing  (at infinity) initial functions.  The authors have developed the inverse scattering method for the both nonlinear equations. For this purpose they gave the formulation and solution of the scattering problems for the self-adjoint and not self-adjoint Dirac operators.

\section{Periodic finite-gap potentials}

Taking into account  ideas of \cite{M}, let us introduce a set of complex-valued $l$-periodic  potentials
\be\label{v0}
v(x,x_0,t):=v(x+x_0,t), \qquad v(x+l,x_0,t)=v(x,x_0,t)
\ee
and related with this set,  the Dirac operators
\be\label{L}
L:= J\frac{d}{dx}+Q, \qquad J=\begin{pmatrix}
                               \ii & 0\\
                               0 & -\ii \\
                             \end{pmatrix}
\ee
with  the not self-adjoint  potential matrix
$$
Q=\begin{pmatrix}
    0 & \ii v(x,x_0,t) \\
    \ii \bar v(x,x_0,t) & 0 \\
  \end{pmatrix}.
$$
The function $v(x,t)$ is supposed to be twice continuously differentiable in $x$ and has the first continuous derivative in $t$.

Let vector $y(x,z)$ is a solution of the equation
\be\label{Lyzy}
Ly=zy,
\ee
where $z$ is an arbitrary complex number.
We introduce additionally two sets of operators:
$$
M_k:=D_k+A_k(x,z,x_0,t), \quad D_1:=\frac{\partial}{\partial t},  \quad D_2:=\frac{\partial}{\partial x_0},
$$
where matrices $A_k$ ($k=1,2$) are choosing in such a way that
$$
(L-zI)M_ky=\{JA_k^\prime+JA_kJ(Q-zI)+(Q-zI)A_k-D_kQ\}y\equiv B_k(x,x_0,t)y.
$$
It is very important that $B_k$ have to be independent on $z$. Such matrices $A_k$ and $B_k$ do exist and have the form:
$$
A_1=\ii\begin{pmatrix}
           2z^2-|v|^2 & v^\prime_x-2\ii zv \\
            \bar v^\prime_x+2\ii z\bar v & -2z^2+|v|^2 \\
         \end{pmatrix},\qquad B_1=-\begin{pmatrix}
                               0 & v^{\prime\prime}_{xx}+2|v|^2v+\ii\dot v_t \\
          -\bar v^{\prime\prime}_{xx}-2|v|^2\bar v+\ii\dot{\bar v}_t  & 0 \\
                             \end{pmatrix}
$$
$$
A_2=\begin{pmatrix}
      \ii z & v \\
      -\bar v & -\ii z \\
    \end{pmatrix}, \qquad B_2=\begin{pmatrix}
                                0 & 0 \\
                                0 & 0 \\
                              \end{pmatrix}.
$$

Now let  $v=v(x,t)$ satisfies equation (\ref{NLS}). Then $B_1\equiv 0$ and, hence, operator $M_k$  transforms a solution of equation (\ref{Lyzy}) into another solution of the same equation (\ref{Lyzy}). Let $\Phi(x,z)\equiv\Phi(x,z,x_0,t)$ be the fundamental matrix of equation (\ref{Lyzy}) and let $\Phi(z,x_0,t)\equiv\Phi(l,z,x_0,t)$ be the monodromy matrix on interval $(0,l)$. Then, if $y(x,z)$ is a solution of equation (\ref{Lyzy}), then $y(x,z)=\Phi(x,z)y(0,z)$. Since operator $M_k$ transforms a solution of equation (\ref{Lyzy}) into another solution of the same equation then
$$
\tilde y_k(x,z):=M_k(x)y(x,z)=M_k(x)\Phi(x,z)y(0,z)
$$
and
$$
\tilde y_k(x,z):=\Phi(x,z)\tilde y_k(0,z)=\Phi(x,z)M_k(0) y(x,z).
$$
Hence
$$
\{M_k(x)\Phi(x,z)-\Phi(x,z)M_k(0) \} y(0,z)=0.
$$
Let us put $x=l$ in the last equation. In view of periodicity of the function $v$ we have $A_k(0,z,x_0,t)\equiv A_k(l,z,x_0,t)$ and therefore the last equation gives the following differential equation for the monodromy matrix:
\be\label{ME}
D_k\Phi(z,x_0,t)=\Phi(z,x_0,t)A_k(l,z,x_0,t)-A_k(l,z,x_0,t)\Phi(z,x_0,t), \qquad k=1,2
\ee
with initial conditions:
$$
\Phi(z,x_0,t)\vert_{t=0}=\Phi_1(z,x_0), \qquad \Phi(z,x_0,t)\vert_{x_0=0}=\Phi_2(z,t),
$$
where $\Phi_1(z,x_0)$ is the monodromy matrix of operators (\ref{L}) with $v=v(x,x_0,0)$, and $\Phi_2(z,t)$ is the monodromy matrix of operators (\ref{L}) with $v=v(x,0,t)$.

Let us put
\begin{align*}
2g=&\Phi_{11}+\Phi_{22},\\
2\ii f=&\Phi_{11}-\Phi_{22},\\
\psi=&\Phi_{12},\\
\varphi=&\Phi_{21},
\end{align*}
where $\Phi_{ik}, i,k=1,2$ are entries of the monodromy matrix $\Phi(z,x_0,t)$.
Then matrix equation (\ref{ME}) gives  $D_k g(z,x_0,t)=0$, i.e. $ g(z,x_0,t)$ is independent on $x_0$ and $t$. Other scalar equations  for functions $f=f(z,x_0,t)$, $\psi=\psi(z,x_0,t)$, $\varphi=\varphi(z,x_0,t)$ are as follows:
\begin{align}\label{tfpsiphi}
\dot f=&(\bar v^\prime+2\ii z\bar v)\psi-( v^\prime-2\ii z v)\varphi,\nonumber\\
\dot\psi=&-2(v^\prime-2\ii z v)f-2\ii(2z^2-|v|^2)\psi,\\
\dot\varphi=&2(\bar v^\prime+2\ii z\bar v)f+2\ii(2z^2-|v|^2)\varphi,\nonumber
\end{align}
and
\begin{align}\label{xfpsiphi}
 f^\prime=& \ii\bar v\psi+ \ii v\varphi ,\nonumber\\
 \psi^\prime=&2\ii  vf-2\ii z\psi,\\
 \varphi^\prime=&2\ii\bar v f+2\ii z \varphi,\nonumber
\end{align}
where the dot and prime mean differentiation in $t$ and $x_0$ respectively, and $v=v(x_0,t)$, $\bar v=\bar v(x_0,t)$.

The following symmetry property takes place: if vector $y(x,z)=(y_1(x,z), y_2(x,z))$ is a solution of (\ref{Lyzy}) then vector $\tilde y(x,z)=(-\overline{y_2(x,\bar z)}, \overline{y_1(x,\bar z)})$ is also a solution to (\ref{Lyzy}). Therefore entries of the monodromy matrix, being entire in $z$ analytic functions of the exponential type $l$, satisfy the following conditions:

$$
\Phi_{22}(z,x_0,t)=\overline{\Phi_{11}(\bar z,x_0,t)}, \qquad
\Phi_{21}(z,x_0,t)=-\overline{\Phi_{12}(\bar z,x_0,t)}.
$$

Hence, entire analytic functions $f$, $\psi$, $\varphi$ possess properties:
\be\label{sym}
f(z,x_0,t)=\overline{f(\bar z,x_0,t)}, \qquad \varphi(z,x_0,t)=-\overline{\psi(\bar z,x_0,t)}.
\ee

Let $\Phi_0(z)$ be the monodromy matrix of equation (\ref{Lyzy}) with independent on $x_0$ and $t$ potential $v=v(x)$ and let corresponding functions $f_0(z)$, $\psi_0(z)$, $\varphi_0(z)$ have the following factorization:
\be\label{h}
f_0(z)=\tilde f_0(z)h(z),\quad \psi_0(z)=\tilde \psi_0(z)h(z),\quad \varphi_0(z)=\tilde \varphi_0(z)h(z),
\ee
where $h(z)$ is an entire analytic function, and $\tilde f_0(z)$, $\tilde \psi_0(z)$, $\tilde \varphi_0(z)$ are some polynomials. One can put the leading coefficient of the polynomial $f_0(z)$ to be equal to one. Such a requirement defines function $h(z)$ uniquely.
\begin{Def*}
The periodic potential is called finite-gap, if the corresponding monodromy matrix $\Phi_0(z)$ possesses properties (\ref{h}).
\end{Def*}
Consider a translation $v(x,x_0)=v(x+x_0)$ of the finite-gap potential $v(x)$.
Then functions $f(z,x_0)$, $\psi(z,x_0)$ and $\varphi(z,x_0)$  defined by the monodromy matrix $\Phi(z,x_0)$ (such that $\Phi(z,0)=\Phi_0(z)$) are the solution of the system (\ref{xfpsiphi}) obeying to initial conditions:
\be\label{fgh0}
f(z,0)=\tilde f_0(z)h(z), \quad \psi(z,0)=\tilde\psi_0(z)h(z),\quad \varphi(z,0)=\tilde\varphi_0(z)h(z)
\ee
\begin{lem}
Under initial conditions (\ref{fgh0}) for equations (\ref{xfpsiphi}) their solution \{$f(z,x_0)$, $\psi(z,x_0)$, $\varphi(z,x_0)$ \}  is represented in the form:
\be\label{fghx0}
f(z,x_0)=\tilde f(z,x_0)h(z), \quad \psi(z,x_0)=\tilde\psi(z,x_0)h(z),\quad \varphi(z,x_0)=\tilde\varphi_0(z,x_0)h(z)
\ee
with the same function $h(z)$, and $\tilde f$, $\tilde\psi$, $\tilde\varphi$ are polynomials precisely the same degrees that the polynomials $\tilde f_0(z)$, $\tilde \psi_0(z)$, $\tilde \varphi_0(z)$. In addition, periodic $n$- gap potential $v(x_0)$ is expressed by the leading coefficient of polynomial $\tilde\psi(z,x_0)=\tilde\psi_n(x_0)z^n+\tilde\psi_{n-1}(x_0)z^{n-1}+\ldots+\tilde\psi_0(x_0)$ by the formula:
\be\label{v}
v(x_0)=\tilde\psi_n(x_0)
\ee
\end{lem}
\begin{proof}
Let $C(z,x_0)$ be the fundamental solution of the matrix equation (\ref{Lyzy})
on the interval $(0,x_0)$. Due to the periodicity of $v(x)$, $C(z,x_0)$ is the fundamental solution of the matrix equation (\ref{Lyzy}) on the interval $(l,l+x_0)$. Then the monodromy matrix $\Phi(z,x_0)$, defined on the interval $(l,l+x_0)$, and the matrix $\Phi_0(z)$, defined on the interval $(0,l)$, are related by equality:
\be\label{C}
\Phi(z,x_0)=C(z,x_0)\Phi_0(z)C^{-1}(z,x_0)
\ee
In view of (\ref{fgh0}) entries of the matrix $\Phi_0(z)$ have the form:
$$
\Phi_0[11]=g+\ii\tilde f_0h,\quad \Phi_0[12]=\tilde\psi_0h,\quad
\Phi_0[21]=\tilde\varphi_0h,\quad \Phi_0[22]=g-\ii\tilde f_0h.
$$
Take into account that $\det C(z,x_0)\equiv 1$ we find from (\ref{C})
\begin{align*}
\Phi_0[11]=&g+\ii\left(\tilde f_0(C[11]C[22]+C[12]C[21])+\ii\tilde \psi_0C[11]C[21]-\ii\tilde\varphi_0C[12]C[22]\right)h,\\
\Phi_0[12]=&\left(\tilde\psi_0C[11]C[11]-\tilde \varphi_0C[12]C[12]-2\ii \tilde f_0C[11]C[12]\right)h,\\
\Phi_0[21]=&\left(\tilde\varphi_0C[22]C[22]-\tilde\psi_0C[21]C[21]+2\ii \tilde f_0C[21]C[22]\right)h,\\
\Phi_0[22]=&g-\ii\left(\tilde f_0)(C[11]C[22]+C[12]C[21])+\ii\tilde \psi_0C[11]C[21]-\ii\tilde\varphi_0C[12]C[22]\right)h
\end{align*}
These equalities give the representations (\ref{fghx0}) with some functions $\tilde f(z,x_0)$, $\tilde\psi(z,x_0)$, $\tilde\varphi(z,x_0)$. They are, in general, are entire analytic. To prove that these functions are polynomials we use rough asymptotic formulas for the entries of the monodromy matrices $\Phi(z,x_0)$ and $\Phi_0(z)$. Thus we find that $\tilde f(z,x_0)$, $\tilde\psi(z,x_0)$, $\tilde\varphi(z,x_0)$ are polynomials of the same degrees that initial ones.
So, if $v(x)$ is a periodic $n$-gap potential then the system of equations (\ref{xfpsiphi}) has the polynomial solution. For  $n$-gap potential the function $\tilde f(z,x_0)$ is a polynomial  of $n+1$-th degree. If so, then equations (\ref{xfpsiphi}) give that $\tilde\psi(z,x_0)$ and $\tilde\varphi(z,x_0)$ are polynomials of $n$-th degrees.
Further, since $\tilde f_{n+1}(x_0)=\tilde f_{n+1}(0)=1$, then the second equation gives equality (\ref{v})
\end{proof}

\section{Autonomous system of equations and solutions of the nonlinear Schr\"odinger  equation}

In this section we show that the  systems  of equations (\ref{tfpsiphi})  and (\ref{xfpsiphi}) generate some nonlinear autonomous equations, which lead to a solution of the nonlinear Schr\"odinger  equation (\ref{NLS}). Let coefficients of equations (\ref{tfpsiphi})  and (\ref{xfpsiphi}) are completely arbitrary, i.e. we will consider the next system of equations:
\begin{align}\label{tfpsiphi1}
\dot f=&(d+2\ii z b)\psi-(c-2\ii z a)\varphi,\nonumber\\
\dot\psi=&-2(c-2\ii z a)f-2\ii(2z^2-ab)\psi,\\
\dot\varphi=&2(d+2\ii zb)f+2\ii(2z^2-ab)\varphi,\nonumber
\end{align}
and
\begin{align}\label{xfpsiphi1}
 f^\prime=& \ii b\psi+ \ii a\varphi ,\nonumber\\
 \psi^\prime=&2\ii  af-2\ii z\psi,\\
 \varphi^\prime=&2\ii b f+2\ii z \varphi,\nonumber
\end{align}
where $a=a(x,t), b=b(x,t), c=c(x,t), d=d(x,t)$ are some arbitrary functions, and the dot and prime mean differentiation in $t$ and $x$ respectively (from here and below we use $x$ instead of $x_0$).

First of all we emphasize that each of the system has the conservation law:
\be\label{P0}
f^2(x,t,z)-\psi(x,t,z)\varphi(x,t,z)\equiv P(z),
\ee
where $P(z)$ is defined by initial conditions. If $f, \psi, \varphi$ satisfy
(\ref{tfpsiphi1})  and (\ref{xfpsiphi1}) simultaneously, then $P(z)$ is independent on $t$ and $x$. Below we consider the case when  $a, b,  c, d$ are unknown functions. In this case the systems of equations are underdetermined (not closed). However, we can make them to be closed. To this end, let us seek polynomial in $z$ solutions of these systems. For arbitrary $n\in\mathbb{N}$, choosing $f$ as a polynomial of $n+1$-th degree
$$
f(x,t,z)=\sum\limits_{j=0}^{n+1}f_j(x,t)z^j,
$$
it is easy to see that   $\psi$ and $\varphi $ have to be polynomials of $n$-th degree:
$$
\psi(x,t,z)=\sum\limits_{j=0}^{n}\psi_j(x,t)z^j, \qquad
\varphi(x,t,z)=\sum\limits_{j=0}^{n}\varphi_j(x,t)z^j
$$
System (\ref{tfpsiphi1}) transforms into the system of equations on polynomial coefficients:
\begin{align}\label{tfpsiphi2}
\dot f_j=&d\psi_j-c\varphi_j+2\ii b \psi_{j-1}+2\ii a\varphi_{j-1},&0\le j\le n+1, \nonumber\\
\dot\psi_l=&-2cf_l +4\ii af_{l-1}+2\ii ab\psi_l - 4\ii\psi_{l-2},&0\le l\le n,\\
\dot\varphi_m=&2df_m+4\ii bf_{m-1}-2\ii ab\varphi_m+4\ii\varphi_{m-2},&0\le m\le n.\nonumber
\end{align}
Here we put $f_j, \psi_l, \varphi_m$ are equal zero when $j, l, m <0$ and $j>n+1$, $l, m >n$. Besides these differential equations we have algebraic relations:
\begin{align*}
4\ii a f_{n+1}-4\ii\psi_{n}=&0,\\
2cf_{n+1}-4\ii a f_n+4\ii\psi_{n-1}=&0,\\
4\ii b f_{n+1}+4\ii\varphi_{n}=&0,\\
2df_{n+1}+4\ii b f_n+4\ii\psi_{n-1}=&0.
\end{align*}
The last relations give unknown functions:
\be\label{ab}
a=\frac{\psi_n}{f_{n+1}},\qquad b=-\frac{\varphi_n}{f_{n+1}},
\ee
\be\label{cd}
c=2\ii\left(\frac{\psi_nf_n}{f^2_{n+1}}-\frac{\psi_{n-1}}{f_{n+1}}\right),\qquad
d=2\ii\left(\frac{\varphi_nf_n}{f^2_{n+1}}-\frac{\varphi_{n-1}}{f_{n+1}}\right).
\ee
Moreover, the first equation of system (\ref{tfpsiphi2}) gives that $f_{n}$ and $f_{n-1}$ is independent on $t$. Substituting $a, b,  c, d$ into (\ref{tfpsiphi2}) we obtain the autonomous system of ODEs:
\be\label{F}
\dot y_j=F_j(y_0, y_1, \ldots, y_N), \qquad j=0,1,\ldots,N, \quad N=3n+1,
\ee
where $y_j=f_j\, (0\le j\le n+1)$, $y_{n+1+j}=\psi_j\, (0\le j\le n)$, $y_{2n+2+j}=\varphi_j\, (0\le j\le n)$, and the right hand sides $F_j$ are at most the third degree polynomials on $y_0, y_1, \ldots, y_N$.

In polynomial coefficients the system (\ref{xfpsiphi1}) is written in the form:
\begin{align}\label{xfpsiphi2}
 f_j^\prime=& \ii b\psi_j+ \ii a\varphi_j, &0\le j\le n+1,\nonumber\\
 \psi_l^\prime=&2\ii  af_l-2\ii\psi_{l-1}, &0\le l\le n,\\
 \varphi_m^\prime=&2\ii b f_m+2\ii\varphi_{m-1}, &0\le m\le n.\nonumber
\end{align}
For the functions $a$ and $b$ formulas (\ref{ab}) are valid  as before. The first equation of the system gives  that $f_{n+1}$ and $f_{n}$ are independent on $x$. After the elimination of $a$ and $b$ these equations can be also written in the form of autonomous system of ODEs:
\be\label{Phi}
y_j^\prime=\Phi_j(y_0,y_1, \ldots, y_N), \qquad j=0,1,\ldots,N, \quad N=3n+1,
\ee
where $\Phi_j$ are at most the second degree polynomials on $y_0, y_1, \ldots, y_N$.

Thus, the requirement of the existence of the polynomial in $z$ solution to systems (\ref{tfpsiphi1}) and (\ref{xfpsiphi1}) defines uniquely the functions $a, b, c, d$ through solutions of autonomous systems (\ref{F}) and (\ref{Phi}). On the other hand, if $f_j$, $\psi_l$ and $\varphi_m$ are a solution of autonomous systems (\ref{F}), (\ref{Phi}), and $a, b, c, d$ are the function constructed by (\ref{ab}), (\ref{cd}) then functions
\be\label{polyn}
f(x,t,z)=\sum\limits_{j=0}^{n+1}f_j(x,t)z^j,\qquad
\psi(x,t,z)=\sum\limits_{j=0}^{n}\psi_j(x,t)z^j, \qquad
\varphi(x,t,z)=\sum\limits_{j=0}^{n}\varphi_j(x,t)z^j
\ee
are a polynomial solution of systems (\ref{tfpsiphi1}), (\ref{xfpsiphi1}).
In what follows the system (\ref{tfpsiphi2}), (\ref{xfpsiphi2}) together with  (\ref{ab}), (\ref{cd}) and the system (\ref{F}), (\ref{Phi}) are identified.

It is well known that necessary and sufficient conditions for the systems
(\ref{F}), (\ref{Phi}) to be compatible are the following ones:
\be\label{Frobenius}
\sum\limits_{j=0}^{N}\left(\frac{\partial F_j}{\partial y_m}\Phi_m-\frac{\partial \Phi_j}{\partial y_m}F_m  \right)=0, \qquad j=0,1,\ldots, N.
\ee
The direct calculation shows that compatibility conditions (\ref{Frobenius}) are fulfilled.
\begin{thm}
Let $f_j(x,t), \psi_l(x,t), \varphi_m(x,t)$ be a compatible local solution of autonomous system (\ref{F}), (\ref{Phi}). Then defined by (\ref{ab}), (\ref{cd}) functions $a(x,t)$ and $b(x,t)$ are a local infinitely differentiable in $x$ and $t$ solution of the following nonlinear  equations:
\be\label{sysNLS}
\ii \dot a+a^{\prime\prime}+2a^2b=0,\qquad
\ii\dot b-b^{\prime\prime}-2b^2a=0.
\ee
\end{thm}

\begin{proof}
Systems (\ref{tfpsiphi2}), (\ref{xfpsiphi2}) contain two pair of equations:
\begin{align}
\dot\psi_n=&-2cf_n+4\ii af_{n-1}+2\ii ab\psi_{n}-4\ii\psi_{n-2},\label{psit}\\
\psi^\prime_n=&2\ii af_n-2\ii\psi_{n-1}\label{psix}
\end{align}
and
\begin{align}
\dot\varphi_n=&2df_n+4\ii bf_{n-1}-2\ii ab\varphi_n+4\ii\varphi_{n-2},\label{phit}\\
\varphi^\prime_n=&2\ii bf_n+2\ii\varphi_{n-1}\label{phix}
\end{align}
Since right hand sides $F_j(y_0,y_1,\ldots,y_N)$ and $\Phi_j(y_0,y_1,\ldots,y_N)$ of systems (\ref{F}), (\ref{Phi}) are infinitely differentiable in $y_0,y_1,\ldots,y_N$ then their solutions are also infinitely differentiable in $t$ and $x$. Differentiating (\ref{psix}) and (\ref{phix}) in $x$, taking into account the  independence of $f_n$  on $x$ and also equations
$$
\psi_{n-1}^\prime=2\ii af_{n-1}-2\ii\psi_{n-2}, \qquad \varphi^\prime_{n-1}=2\ii b\varphi_{n-1}+2\ii\varphi_{n-2}
$$
we find
$$
\psi_n^{\prime\prime}=2\ii a^\prime f_n+4af_{n-1}-4\psi_{n-2},\qquad
\varphi^{\prime\prime}_n=-2\ii b^\prime f_n+4bf_{n-1}+4\varphi_{n-2}
$$
Let us multiply the last equations on $\ii$ and subtract them from (\ref{psit}) and (\ref{phit}). Then we have
$$
\dot\psi_n-\ii\psi_n^{\prime\prime}-2\ii ab \psi_n=2(a^\prime-c)f_{n-1}, \qquad
\dot\varphi_n+\ii\varphi_n^{\prime\prime}+2\ii ab \varphi_n=2(d-b^\prime)f_{n-1}.
$$
As $f_{n+1}$ is independent on $t$ and $x$ then formulas (\ref{ab}), (\ref{cd}) together with equations (\ref{psix}), (\ref{phix}) give that $a^\prime=c$ and $d=b^\prime$. Finally, since $\psi_n=af_{n+1}$ and $\varphi_n=-bf_{n+1}$, then the last two equations coincide with system (\ref{sysNLS}).
\end{proof}

Under some conditions the systems (\ref{tfpsiphi2}), (\ref{xfpsiphi2}) have a global solution.
\begin{lem}
Let initial conditions satisfy the symmetry properties:
\be\label{sym1}
f_j=\bar f_j, \qquad \varphi_j=-\bar\psi_j.
\ee
Then system  (\ref{tfpsiphi2}), as well as  the system (\ref{xfpsiphi2}), has a solution defined for all $t\in\mathbb{R}$ or $x\in\mathbb{R}$ respectively.  The solution possesses the same property (\ref{sym1}) for all $t$ and $x$.
\end{lem}

\begin{proof}
The proof will be done for system (\ref{tfpsiphi2}).  We first show that the symmetry property is conserved for all $t$. Indeed, let $f_j(t)$, $\psi_l(t)$, $\varphi_m(t)$ be a solution of (\ref{tfpsiphi2}). Consider a new functions:
$\hat f_j(t)=\bar f_j(t), \, \hat\psi_l(t)=-\bar\varphi_l(t), \, \hat\varphi_m(t)=-\bar\psi_m(t)$. After complex conjugation of system (\ref{tfpsiphi2}) it is easy to see that $\hat f_j(t), \, \hat\psi_l(t), \, \hat\varphi_m(t)$ satisfy (\ref{tfpsiphi2}). Under condition $t=t_0$, in view of (\ref{sym1}), these two sets of functions are coincide. Then the uniqueness theorem for ODEs gives that they are coincide for all $t$. Further, the existence theorem for ODEs gives the local solution of  (\ref{tfpsiphi2}), which inherits the symmetry properties (\ref{sym1}). In turn these properties lead to the uniform boundedness of the solution and, hence, the solution has an extension for all $t\in\mathbb{R}$.

The uniform boundedness  is a consequence   of the conservation law (\ref{P0}):
\be\label{P1}
P(z)\equiv f^2(z,t)-\psi(z,t)\varphi(z,t),
\ee
where polynomials $f, \psi, \varphi$ are constructed by $f_j(t)$, $\psi_l(t)$, $\varphi_m(t)$. In view of (\ref{sym1}), $f(z,t)=\overline f(\bar z,t)$ and $\varphi(z,t)=-\overline\psi(\bar z,t)$. Hence $P(z)=\overline P(\bar z)$. Therefore the polynomial $P(z)$ can be factorized:
$$
P(z)=P^+(z)P^-(z),
$$
where $P^+(z)$ is a polynomial of $n+1$-th degree. Its zeroes are located in the closed upper half plane. The second polynomial is the complex conjugated to the first one: $P^-(z)=\overline P^+(\bar z)$. For real $z$ the conservation law (\ref{P1}) takes the form
$$
f^2(z,t)+|\psi(z,t)|^2\equiv|P^+(z)|^2.
$$
Hence
$$
|f(z,t)|\le |P^+(z)|, \qquad |\psi(z,t)|=|\varphi(z,t)|\le |P^+(z)|,\qquad z\in\mathbb{R}.
$$
By using the Bernstien inequalities \cite{Levin} we find
$$
\left\vert\frac{\partial^m f(z,t)}{\partial z^m}\right\vert\le
\left\vert\frac{d^m P^+(z)}{dz^m}\right\vert, \quad m=1,2,\ldots, \quad z\in\mathbb{R}.
$$
Choosing $m=0,1,2,\ldots, n$ and $z=0$ we obtain
$$
|f_m(t)|\le|P^+_m|, \qquad |\psi_m(t)|=|\varphi_m(t)|\le |P^+_m|,\qquad
m=1,2,\ldots,n+1,
$$
where $P^+_m$ are the coefficients of polynomial $P^+(z)$, which are uniquely defined by initial conditions of the  autonomous system.  The analogous proof goes for the system (\ref{xfpsiphi2}).
\end{proof}
\begin{cor}
If initial conditions for the systems (\ref{tfpsiphi2}) and (\ref{xfpsiphi2}) possess properties (\ref{sym1}) then there exists  the compatible global infinitely differentiable and bounded solution $\hat f_j(t)=\bar f_j(t), \, \hat\psi_l(t)=-\bar\hat\varphi_l(t), \, \hat\varphi_m(t)=-\bar\psi_m(t)$.
Moreover, the function $a(x,t)=\psi_n(x,t)/f_{n+1}(0,0)$ is a solution of the nonlinear Schr\"odinger  equation (\ref{NLS}).
\end{cor}

Indeed, due to the symmetry properties,
$$
b(x,t)=-\frac{\varphi_n(x,t)}{f_{n+1}(0,0)}=\frac{\overline\psi_n(x,t)}{f_{n+1}(0,0)}=
\overline a(x,t).
$$
Therefore the first equation in (\ref{sysNLS}) is the NLS equation (\ref{NLS}), and the second is   complex-conjugated with the first (\ref{NLS}).

\section{Periodic Cauchy problem for the nonlinear Schr\"odiger  equations with
finite-gap initial data}

On the whole line $-\infty<x<\infty$ let us consider the Cauchy periodic problem for equation (\ref{NLS}), whose an initial datum $u(x)$  is   $l$-periodic $n$-gap potential. Consider also the set of shifts $v(x,x_0)=u(x+x_0)$. Then, due to the Lemma 1, the monodromy matrix $\Phi(z,x_0)$ and corresponding functions $f(z,x_0), \psi(z,x_0), \varphi(z,x_0)$ have the form (\ref{h}) with a function $h(z)$ that independent on $x_0$,  and  $\tilde f(z,x_0), \tilde\psi(z,x_0), \tilde\varphi(z,x_0)$ are a polynomial solution of equation (\ref{xfpsiphi}).
It follows from the periodicity of the monodromy matrix $\Phi(z,x_0)$, the polynomials $\tilde f(z,x_0), \tilde\psi(z,x_0), \tilde\varphi(z,x_0)$ are $l$-periodic functions on $x_0$.
Therefore coefficients of the polynomials $\tilde f(z,x_0), \tilde\psi(z,x_0), \tilde\varphi(z,x_0)$ are a periodic solution of autonomous system (\ref{xfpsiphi2}).
Further we solve (for every fixed $x_0$) the system (\ref{tfpsiphi2}) with initial data that is defined by the coefficients of  polynomials $\tilde f(z,x_0), \tilde\psi(z,x_0), \tilde\varphi(z,x_0)$. Taking into account the Corollary from Theorem 1 we obtain some solution $v(x_0,t)$ of equation (\ref{NLS}), where we identify $x_0$ with $x$. The general theory of ODEs guarantees that obtained solution $v(x_0,t)$ is $l$-periodic smooth  and unique.
Obtained simultaneously the polynomials  $\tilde f(z,x_0,t), \tilde\psi(z,x_0,t), \tilde\varphi(z,x_0,t)$ are the solution of linear system (\ref{tfpsiphi1}) because from now  the function $v(x_0,t)$ is already known. After a multiplication on the function $h(z)$ we obtain once more solution of  (\ref{tfpsiphi1}) with the initial datum generated by potential $u(x)$. On the other hand, the solution $v(x_0,t)$ of equation (\ref{NLS}) defines monodromy matrix and corresponding the entire analytic functions $ f(z,x_0,t),  \psi(z,x_0,t),  \varphi(z,x_0,t)$, which satisfy the system  (\ref{tfpsiphi1}) with the  initial data generated by $u(x)$. Again, the uniqueness theorem  gives that
\begin{align*}
f(z,x_0,t)=&\tilde f(z,x_0,t) h(z),\\
\psi(z,x_0,t)=&\tilde\psi(z,x_0,t) h(z),\\
\varphi(z,x_0,t)=&\tilde\varphi(z,x_0,t)h(z).
\end{align*}
The last equalities mean that $v(x_0,t)$ are $n$-gap potential  for every $t$.
Thus we proved
\begin{thm}
For any periodic $n$-gap initial datum, the Cauchy periodic problem to the Schr\"odinger nonlinear equation (\ref{NLS}) is uniquely solvable. Moreover, the solution is also periodic $n$-gap potential for every fixed $t$.
\end{thm}

\section{Dynamics of zeroes of polynomials}

Consider the compatible solution $f_j(x,t), \psi_l(x,t), \varphi_m(x,t)$ ($0\le j\le n,\, 0\le l,m\le n-1$) of autonomous systems (\ref{tfpsiphi1}) and (\ref{xfpsiphi1}), whose initial data have the properties (\ref{sym1}). Then, due to the previous considerations, formula
\be\label{v1}
v(x,t)=a(x,t)=\frac{\psi_n(x,t)}{f_{n+1}(0,0)}
\ee
defines some solution of nonlinear equation (\ref{NLS}). Moreover, constructed by (\ref{polyn}) polynomials $f(z,x,t), \psi(z,x,t), \varphi(z,x,t) $ satisfy  properties (\ref{sym}). They are the solution of linear systems (\ref{tfpsiphi}) and (\ref{xfpsiphi}). It allow to give another representation for the solution $v(x,t)$. Indeed, let $\mu_j=\mu_j(x,t)$ ($j=1,2,\ldots,n$) be zeroes of the polynomial $\psi(z,x,t)$. Zeroes of the polynomial $\varphi(z,x,t)$ are complex- conjugated to $\mu_j$. The Vieta formulas give
\be\label{smu}
\sum\limits_{j=1}^n\mu_j(x,t)=-\frac{\psi_{n-1}(x,t)}{\psi_n(x,t)},
\ee
\be\label{snu}
\sum\limits_{j=1}^n\nu_j(x,t)=-\frac{f_{n}(0,0)}{f_{n+1}(0,0)},
\ee
where $\nu_j(x,t)$ ($j=1,2,\ldots,n$) are zeroes of polynomial $f(z,x,t)$.
Without loss of generality we  put $f_{n+1}(0,0)=1$.
Then, due to conservation law (\ref{P1}),
$$
\frac{f_{n}(0,0)}{f_{n+1}(0,0)}=-\frac{1}{2}\sum\limits_{j=1}^{n+1}( E_j+ \bar E_j)=K,
$$
where $E_j$ and $\bar E_j$ are zeroes of polynomial
\be\label{P2}
P(z)=\prod\limits_{j=1}^{2n+2} (z-E_j)\equiv f^2(z,x,t)-\psi(z,x,t)\varphi(z,x,t).
\ee
The second equation of system (\ref{xfpsiphi1}) with $l=n$ together with (\ref{smu}), (\ref{snu}) give
\be\label{vx}
\frac{\partial}{\partial x}\ln\psi_n(x,t)=\varPhi(x,t),\qquad \varPhi(x,t) =2\ii \sum\limits_{k=1}^{n}\mu_k(x,t)+2\ii K.
\ee
By analogy, the second equation of system (\ref{tfpsiphi1}) with $l=n$ together with (\ref{P2}), (\ref{smu}), (\ref{snu}) give
\be\label{vt}
\frac{\partial}{\partial t}\ln\psi_n(x,t) =F(x,t),
\ee
where
$$F(x,t)=2\ii\left[ \sum\limits_{j>k} E_jE_k-\frac{3}{4}\left( \sum\limits_{k=1}^{2n+2}E_k\right)^2\right]-
4\ii\left[K\sum\limits_{k=1}^{n}\mu_k(x,t)+\sum\limits_{j>k}\mu_j(x,t)\mu_k(x,t)\right]. $$
Here $\mu_j(x,t)$ are zeroes of polynomial $\psi(z,x,t)$, and $E_j$ are zeroes of the polynomial $P(z)$. An integration of (\ref{vx}), (\ref{vt}) lead to the representation for $v(x,t)$:
$$
v(x,t)=\psi_n(0,0) \exp\left\{\int\limits_0^t f(0,s)ds + \int\limits_0^x \Phi(\xi,t)d\xi \right\}.
$$

The systems of equations (\ref{tfpsiphi1}), (\ref{xfpsiphi1}) gives differential equations for $\mu_j(x,t)$. Indeed, dividing the second equation of the system (\ref{tfpsiphi1}) on $\psi$ and using (\ref{ab}), (\ref{cd}) we find
$$
\frac{\partial}{\partial t}\ln\psi(z,x,t)=\frac{\partial}{\partial t}\ln\left(
\psi_n\prod\limits_{l=1}^n(z-\mu_l)\right)=-4\ii v \left(\sum\limits_{k=1}^{n}\mu_k(x,t) +K-z\right)\frac{f(z)}{\psi(z)}.
$$
 Here and below we suppose that zeroes $\mu_j$ are simple. The multiplication of this equation on $z-\mu_j$ and the passage (as $z\to\mu_j$) to the limit give the following equations:
$$
\frac{\partial}{\partial t}\mu_j(x,t)=
4\ii v\left(\sum\limits_{l=1}^{n}\mu_l+K-\mu_j\right)
\frac{f(\mu_j)}{\frac{\partial\psi}{\partial z}(\mu_j)}.
$$
Since
$$
\psi(z)=\psi_n \prod\limits_{l=1}^n(z-\mu_l)
$$
then
$$
\frac{\partial\psi(z)}{\partial z}\vert_{z=\mu_j}= \psi_n \prod\limits_{l\neq j}^n(\mu_j-\mu_l).
$$
Finally, since $f^2(\mu_j)=P(\mu_j)$ and $v=\psi_n$, then
\be\label{mut}
\frac{\partial}{\partial t}\mu_j=4\ii \left(\sum\limits_{l=1}^{n}\mu_l+K-\mu_j\right)
\frac{\sqrt{P(\mu_j)}}{\prod\limits_{l\neq j}(\mu_j-\mu_l)},\qquad j=1,2,\ldots,n
\ee
By the same way we find
\be\label{mux}
\frac{\partial}{\partial x}\mu_j(x,t)= -2\ii\frac{\sqrt{P(\mu_k)}}{\prod\limits_{l\neq j}(\mu_j-\mu_l)},\qquad j=1,2,\ldots,n.
\ee
Identity (\ref{P2}) means that if an  arbitrary polynomial $P(z)$ is positive on the real axis then initial values $\mu_j(0,0)$ and $\psi_n(0,0)$ as parameters have to satisfy the condition:
$$
\sum\limits_{j=1}^{n+1}(z- E_j)(z-\bar E_j)
-|\psi_n(0,0)|^2\prod\limits_{j=1}^{n}(z-\mu_j(0,0))(z-\bar\mu_j(0,0))=f^2(z),
$$
where $f(z)$ is a polynomial of $n$-th degree with real coefficients.

Earlier  such  systems, like (\ref{mut}),  (\ref{mux}), and related to the Korteweg de Vries equation, were  considered in \cite{M}, \cite{IM}, \cite{D}. In \cite{IM}, \cite{D} they were integrated with the help of  the Abel map and the Jacobi inversion problem. Our systems (\ref{mut}),  (\ref{mux}) are also integrated by the same procedure. Below we reproduce the corresponding result \cite{IK}.
\\\\

\centerline{A.R.Its and V.P. Kotlyarov}

\centerline{\bf Explicit formulas for solutions of the nonlinear Schr\"odinger  equation }
\bigskip
\centerline{(Presented by the academician V.A.Marchenko 11.XII.1975) }
\bigskip
\centerline{\bf Summary }

{\it The explicit formulas are obtained for solutions of the Schr\"odinger nonlinear equation $$\ii u_t+u_{xx}+2|u|^2u=0.$$ The formulas are constructed by means of $\theta$-functions.}
\bigskip
\setcounter{section}{6}
\setcounter{equation}{0}

The study  of periodic and almost periodic solutions of some nonlinear partial differential equations has shown a relationship of such  solutions with classical objects of the algebraic geometry \cite{DMN}. This gave the possibility to construct explicit formulas for solutions of some nonlinear equations by using theta functions. The present report is devoted to the construction of such type of solutions to the equation:
\be\label{NLS1}
\ii u_{t}+u_{xx}+ 2|u|^2 u=0.
\ee
Below we will use results and methods of \cite{K} and \cite{I}.

Let coefficients $f_k(x,t), \psi_l(x,t), \varphi_m(x,t)$ ($0\le k\le n+1, 0\le l,m\le n$) of polynomials $f(z,x,t), \psi(z,x,t), \varphi(z,x,t) $ satisfy following conditions: $f_k=\bar f_k, \psi_m=-\bar\varphi_m$. It is proved in \cite{K} that, if functions $f_k(x,t), \psi_l(x,t), \varphi_m(x,t)$ are a compatible solution to the special autonomous system of differential equations, then function $u(x,t):=\psi_{n}(x,t)/f_{n+1}(0,0)$ is a bounded infinitely differentiable solution to equation (\ref{NLS1}). Moreover, the following formulas hold:
\be\label{ux}
\frac{\partial}{\partial x}\ln u(x,t)=2\ii \sum\limits_{k=1}^{n}\mu_k(x,t)+2\ii K, \qquad K=-\frac{1}{2}\sum\limits_{k=1}^{2n+2} E_k,
\ee
\be\label{ut}
\frac{\partial}{\partial t}\ln u(x,t)=2\ii\left[ \sum\limits_{j>k} E_jE_k-\frac{3}{4}\left( \sum\limits_{k=1}^{2n+2}E_k\right)^2\right]-
4\ii\left[K\sum\limits_{k=1}^{n}\mu_k(x,t)+\sum\limits_{j>k}\mu_j(x,t)\mu_k(x,t)\right], \ee
where $\mu_j(x,t)$ are zeroes of polynomial $\psi(z,x,t)$, and $E_j$ are zeroes of the polynomial
\be\label{P6}
P(z)=\prod\limits_{j=1}^{2n+2} (z-E_j)\equiv f^2(z,x,t)-\psi(z,x,t)\varphi(z,x,t),
\ee
which is positive on the real axis. Without loss of generality we have put $f_n(0,0)=1$. We emphasize that  $f^2-\psi\varphi$ is independent on $x$ and $t$. Therefore complex numbers $E_j$ are also independent on $x$ and $t$. The functions $\mu_j(x,t)$ satisfy the following equations:
\be\label{mux1}
\frac{\partial}{\partial x}\mu_k(x,t)=-2\ii\frac{\sqrt{P(\mu_k)}}{\prod\limits_{j\neq k}(\mu_k-\mu_j)},
\qquad k=1,2,\ldots,n;
\ee
\be\label{mut1}
\frac{\partial}{\partial t}\mu_k(x,t)=4\ii\left(\sum\limits_{j=1}^{n}\mu_j+K-\mu_k\right)
\frac{\sqrt{P(\mu_k)}}{\prod\limits_{j\neq k}(\mu_k-\mu_j)}.
\ee
Equations (\ref{mux1}) and (\ref{mut1}) have to be considered on the Riemann surface $\sigma$ of the function $\sqrt{P(z)}$. The upper or lower sheet of the Riemann surface is defined by the sign of the square root. Point $\mu_k(x,t)$ can be located only on  one  sheet  that is defined uniquely by equality: $f(\mu_k(x,t),x,t)=\sqrt{P(\mu_k(x,t))}$. We will use in the usual  way (sf.\cite{I}) a  realization of the Riemann surface $\sigma$ and basis $(a_j, b_j)$ of cycles on it. The Abel map, after applying to equations (\ref{mux1}), (\ref{mut1}), gives that points $\mu_k(x,t)$ are a solution of  the hyper-elliptic Jacobi inversion problem of Abelian integrals:
\be\label{JIP}
\sum\limits_{k=1}^{n} U_j(\mu_k(x,t))=-2\ii C_j^1x-4\ii(C_j^2-KC_j^1)+
\sum\limits_{k=1}^{n} U_j(\mu_k(0,0)),
\ee
where Abelian integrals
$$
U_j(z)=\int\limits_{E_{2n+2}}^z(C_j^1\la^{n-1}+C_j^2\la^{n-2}+\ldots+C_j^{n})
\frac{d\la}{\sqrt{P(\la)}}, \qquad j=1,2,\dots, n
$$
are normalized by conditions:
$$
\int\limits_{a_k} dU_j(z)=\delta_{kj}, \qquad k,j=1,2,\dots, n.
$$
Equations (\ref{JIP}) give the following relations:
\be\label{summu}
\sum\limits_{k=1}^{n} \mu_k(x,t)= K_1-\frac{\ii}{2}\frac{\partial}{\partial x}\ln J_1(x,t), \qquad K_1=\sum\limits_{k=1}^{n}\int\limits_{a_k}z d U_k(z),
\ee
\be\label{summu2}
\sum\limits_{k=1}^{n} \mu^2_k(x,t)= K_2-\frac{\ii}{4}\frac{\partial}{\partial t}\ln J_1(x,t)+\frac{1}{4}\frac{\partial^2}{\partial x^2}\ln J_2(x,t), \qquad K_2=\sum\limits_{k=1}^{n}\int\limits_{a_k}z^2 d U_k(z),
\ee
where
\begin{align*}
J_1(x,t)=&\frac{\theta(g(x,t)-r)}{\theta(g(x,t)+r)},\\
J_2(x,t)=&\theta(g(x,t)-r)\theta(g(x,t)+r),\\
\theta(p)=&\sum\limits_{m\in\mathbb{Z}^{n}}\exp[\pi\ii(Bm,m)+2\pi\ii(p,m)], \quad g(x,t),r\in\mathbb{C}^{n},\\
g_j(x,t)=&-2\ii C_j^1 x-4\ii(C_j^2-KC_j^1)t+\sum\limits_{k=1}^{n} U_j(\mu_k(0,0))+
\frac{1}{2}\left(\sum\limits_{k=1}^{n} B_{kj}-j  \right),\\
r_j=&\int\limits_{E_{2n+2}}^{\infty^+} d U_j(z),\qquad B_{kj}=\int\limits_{b_k}
d U_j(z),
\end{align*}
$\infty^+$ infinite point on the Riemann surface $\sigma$.

Substituting (\ref{summu}) and (\ref{summu2}) into (\ref{ux}) and (\ref{ut}) one finds
\be\label{mux2}
\frac{\partial}{\partial x}\mu_k(x,t)=\frac{\partial}{\partial x}\ln J_1(x,t)-\ii E, \qquad E=\sum\limits_{k=1}^{2n+2} E_j-2K_1,
\ee
\be\label{mut2}
\frac{\partial}{\partial t}\mu_k(x,t)=\frac{\partial}{\partial t}\ln J_2(x,t)+\ii N(x,t),
\ee
where
$$
N(x,t)=\frac{\ii}{2}\frac{\partial}{\partial t}\ln J_1(x,t)+\frac{1}{2}\frac{\partial^2}{\partial x^2}\ln J_2(x,t)+\frac{1}{2}\left(\frac{\partial}{\partial x}\ln J_1(x,t)-\ii E\right)^2+2K_2-\sum\limits_{k=1}^{2n+2} E_j^2.
$$
Equations (\ref{mux2}) and (\ref{mut2}) give that the                                                                                                              function $N(x,t)$ is independent of  $x$. In principal, equations (\ref{mux2}) and (\ref{mut2}) give already the desired explicit formula for the solutions equation (\ref{NLS1}). This formula can be essentially improved. To this end let us show that
\be\label{N}
N(x,t)\equiv N_0:=4K_2-2\sum\limits_{k=1}^{2n+2} E_j^2-4R_1,
\ee
and
\be\label{|u|2}
|u(x,t)|^2=\frac{\partial^2}{\partial x^2}\ln \theta(g(x,t)+r)+2R_1,
\ee
where $R_1$ is a constant value. This value depends only on numbers $E_j$. We stress that (\ref{|u|2} is not a consequence of equations (\ref{mux2}) and (\ref{mut2}).  Below we give a sketch of deducing of equalities (\ref{N}) and (\ref{|u|2}).

Let us consider the operator
$$
L:=\begin{pmatrix}
     \ii & 0 \\
     0 & -\ii \\
   \end{pmatrix}\frac{d}{dx}+\begin{pmatrix}
                               0 & \ii u \\
                               \ii \bar u & 0 \\
                             \end{pmatrix},
$$
where $u=u(x,t)$ is defined by equations (\ref{mux2}) and (\ref{mut2}). It can be shown that equation $Ly=zy$ has such a solution $\psi(x,z)$ which is single valued and analytic (in $z$) on the Riemann surface $\sigma$, and its first component  $\psi_1(x,z)$ satisfies the following conditions:
\begin{enumerate}
\item zeroes of $\psi_1(x,z)$ are simple and lie over points $\mu_k(x,t)$;\\
\item poles of $\psi_1(x,z)$ are simple and lie over points $\mu_k(0,t)$;\\
\item $\psi_1(x,z)\sim e^{\ii zx}\D\frac{u(x,t)}{u(0,t)}, \quad z\to\infty^+$,\qquad on the upper sheet of $\sigma$;\\
\item $\psi_1(x,z)\sim e^{-\ii zx},\quad z\to\infty^-$,\qquad on the lower sheet of $\sigma$.
\end{enumerate}
The existence of solution $\psi(x,z)$ to equation $Ly=zy$  is already sufficient for deducing formulas (\ref{N}) and (\ref{|u|2}). Indeed, the properties i-iv lead to a representation for the function

$$
\varphi(x,z):=[u(x,t)]^{-1/2}[u(0,0)]^{1/2}\psi_1(x,z)
$$
through the theta functions:
\be\label{phi}
\varphi(x,z)=e^{\ii x\omega(z)}\frac{\theta(U(z)-g(x,t))}{\theta(U(z)-g(0,t))}
\left[\frac{\theta(g(0,t)-r)\theta(g(0,t)+r)}{\theta(g(x,t)-r)\theta(g(x,t)+r)} \right]^{1/2},
\ee
where $\omega(z)$ is the normalized Abelian integral of the second kind with simple poles at $\infty^\pm$:
$$
\omega(z)=\pm\left(z-\frac{E}{2}+\frac{R_1}{z}+\frac{R_2}{z^2}+\ldots  \right).
$$
 Numbers $R_1$, $R_2$ as well as $E$ are defined via   complex numbers $E_j$.
Moreover, by the definition, the function $\varphi(x,z)$ satisfies equation
$$
\frac{\partial^2\varphi(x,z)}{\partial x^2}+\left[z^2-\ii z \frac{\partial}{\partial x}\ln u(x,t)+\frac{1}{2}\frac{\partial^2 }{\partial x^2}\ln u(x,t)-\frac{1}{4}\left(\frac{\partial}{\partial x}\ln u(x,t) \right)^2 +|u(x,t)|^2 \right]\varphi(x,z)=0.
$$
As follows from \cite{I}, the last equation together with (\ref{phi}) lead to equality (\ref{|u|2}). Substituting  $2\ii |u|^2+\ii\frac{\partial^2 }{\partial x^2}\ln u+\ii\left(\frac{\partial}{\partial x}\ln u \right)^2$ into the expression for $N(x,t)$ from (\ref{mut2}), using (\ref{|u|2}) for $|u(x,t)|^2$ and  (\ref{mux2}) for $\frac{\partial}{\partial x}\ln u$ we arrive to (\ref{N}).

Thus from (\ref{mux2}) and (\ref{mut2}) we have the final formula for solutions $u(x,t)$ equation (\ref{NLS1}):
\be\label{finu}
u(x,t)=u(0,0)\frac{J_1(x,t)}{J_1(0,0)}\exp{(-\ii Ex+\ii N_0t)}.
\ee
Note that formula (\ref{|u|2}) is true for the modulus of $u(x,t)$.
From (\ref{P6}) follows that solution (\ref{finu}) is completely defined by parameters connected by a polynomial relation:
$$
\sum\limits_{j=1}^{n+1}(z- E_j)(z- \bar E_j)
-|u(0,0)|^2\prod\limits_{j=1}^{n}(z-\mu_j(0,0))(z-\bar\mu_j(0,0))=f^2(z),
$$
where $f(z)$ is a polynomial of $n+1$-th degree with real coefficients.


\end{document}